\documentclass{jfrarticle}
\usepackage{afterpage}

\usepackage{amsthm}
\usepackage{amsmath}
\usepackage{lastpage}
\usepackage{url}
\newcommand\articlecommand[1]{{\textsf{#1}}}
\newcommand{\inference}[3]{\ensuremath{\frac{~#2~}{~#3~}{\mbox{~\articlecommand{#1}}}}}
\newcommand{\theory}[2]{{#1}\,\triangleright\,{#2}}

\theoremstyle{definition}
\newtheorem{definition}{Definition}[section]
\theoremstyle{theorem}
\newtheorem{theorem}{Theorem}[section]
\theoremstyle{theorem}
\newtheorem{lemma}{Lemma}[section]
\theoremstyle{remark}
\newtheorem{remark}{Remark}[section]
\DeclareMathOperator{\NF}{NF}
\DeclareMathOperator{\TD}{TD}
\DeclareMathOperator{\Ch}{Ch}
\DeclareMathOperator{\toWff}{toWff}
\DeclareMathOperator{\toBool}{toBool}

\title{Conversion of HOL Light proofs into Metamath}

\author{Mario Carneiro}
{MARIO M. CARNEIRO\\Department of Mathematics, Ohio State University, Columbus OH 43210, USA}
\begin{abstract}
We present an algorithm for converting proofs from the OpenTheory interchange format, which can be translated to and from any of the HOL family of proof languages (HOL4, HOL Light, ProofPower, and Isabelle), into the ZFC-based Metamath language. This task is divided into two steps: the translation of an OpenTheory proof into a Metamath HOL formalization, \texttt{hol.mm}, followed by the embedding of the HOL formalization into the main ZFC foundations of the main Metamath library, \texttt{set.mm}. This process provides a means to link the simplicity of the Metamath foundations to the intense automation efforts which have borne fruit in HOL Light, allowing the production of complete Metamath proofs of theorems in HOL Light, while also proving that HOL Light is consistent, relative to Metamath's ZFC axiomatization. 
\end{abstract}

\begin{document}

%sets the number of the first page
%\setcounter{page}{111} 

\begin{bottomstuff}
\permission
\copyright\ 2015 Journal of Formal Reasoning
\end{bottomstuff}

\maketitle

\section{Introduction}

Metamath is a proof language, developed in 1992, on the principle of minimizing the foundational logic to as little as possible \cite{metamath}. The resulting logic has only one built-in rule of inference, direct substitution, and all syntax and axioms are input outside the logical core. The most well-developed axiom system in Metamath is called \texttt{set.mm}, and adds the axioms of classical propositional calculus, first-order predicate calculus, and ZFC set theory to the Metamath foundations; this axiom system and associated library of theorems is also sometimes referred to as Metamath. In contrast, OpenTheory is an interchange format for the HOL family of proof languages (HOL4, HOL Light, ProofPower/HOL, and Isabelle/HOL) based on a higher-order logical kernel \cite{otbook}. All the axioms of HOL are built into the kernel, and several axioms perform proper substitution, which involves the renaming of bound variables, as part of their operation. The goal of this paper is to present an algorithm that transforms valid theorem derivations in OpenTheory into equivalent theorem derivations in \texttt{set.mm}, as a roadmap for an eventual implementation.

The main task divides neatly into two parts. First, the axioms and inferences of OpenTheory are translated into their Metamath equivalents, producing a new database of axioms which we will call \texttt{hol.mm}. At this stage it is still essentially a HOL derivation, but the inferences are done ``the Metamath way.'' The primary job here is to eliminate proper substitutions and dummy variable renaming, which must be done over several steps in Metamath, and introduce metavariables in place of free variables in the final statement.

The second task is to convert a \texttt{hol.mm} derivation into a \texttt{set.mm} derivation. This step is done entirely within Metamath, and works by constructing a model of HOL within ZFC. Types become sets, functions become ZFC functions (sets of ordered pairs), and the indefinite descriptor becomes a choice function on the HOL universe.

\section{Part I: Conversion from OpenTheory to \MakeLowercase{\texttt{hol.mm}}}

In this part, we present a transformation from OpenTheory article file format to Metamath \texttt{hol.mm} format.

\subsection{OpenTheory}
\label{sec:opentheory}

Although the original goal of this research was the translation of HOL Light proofs to Metamath, the source of the translation was changed to OpenTheory because it already provides bidirectional translation to HOL Light. Furthermore, the OpenTheory file format is much simpler than the HOL Light file format, requiring only minimal work to read the derivations into an article reader. But more importantly, the derivations are all laid out already without needing to be generated first, since HOL Light ``proofs'' are not actually derivations but instructions for searching for derivations. By targeting OpenTheory we are able to skip the step of deriving the proof to begin with and start from a baseline of a complete proof in HOL-compatible format.

\afterpage{
\begin{figure}
\begin{center}
\framebox[\hsize][c]{
  \vbox{
    \vspace{2mm}
    \hbox to\hsize{
      \hspace{2mm}
      \(\displaystyle\inference{varTerm $v$\footnotemark[1]}{}{v:\alpha}\)
      \hfill
      \(\displaystyle\inference{absTerm $v$\footnotemark[1]}
          {t:\beta}{(\lambda v.\ t):\alpha\to\beta}\)
      \hfill
      \(\displaystyle\inference{appTerm}{f:\alpha\to\beta\quad x:\alpha}{f\;x:\beta}\)
      \hspace{2mm}
    }
    \vspace{5mm}
    \hbox to\hsize{
      \hspace{2mm}
      \(\displaystyle\inference{refl}{t:\alpha}{\vdash t = t}\)
      \hfill
      \(\displaystyle\inference{assume}{\phi:{\sf bool}}{\{\phi\}\vdash\phi}\)
      \hfill
      \(\displaystyle\inference{eqMp\footnotemark[2]}
          {\Gamma\vdash\phi'=\psi\quad\Delta\vdash\phi}
          {\Gamma\cup\Delta\vdash\psi}\)
      \hspace{2mm}
    }
    \vspace{5mm}
    \hbox to\hsize{
      \hspace{2mm}
      \(\displaystyle\inference{appThm}
          {f:\alpha\to\beta\quad x:\alpha\quad
            \Gamma\vdash f = g\quad\Delta\vdash x = y}
          {\Gamma\cup\Delta\vdash f\;x = g\;y}\)
      \hfill
      \(\displaystyle\inference{absThm $v$\footnotemark[3]}{\Gamma\vdash t = u}
          {\Gamma\vdash(\lambda v.\ t) = (\lambda v.\ u)}\)
      \hspace{2mm}
    }
    \vspace{5mm}
    \hbox to\hsize{
      \hspace{2mm}
      \(\displaystyle\inference{deductAntisym}
          {\Gamma\vdash\phi\quad\Delta\vdash\psi}
          {(\Gamma-\{\psi\})\cup(\Delta-\{\phi\})\vdash\phi=\psi}\)
      \hfill
      \(\displaystyle\inference{betaConv}
          {((\lambda v.\ t)\;u):\alpha}{\vdash(\lambda v.\ t)\;u = t[v\to u]}\)
      \hspace{2mm}
    }
    \vspace{5mm}
    \hbox to\hsize{
      \hspace{2mm}
      \(\displaystyle\inference{subst $\sigma$}
          {\Gamma\vdash\phi}
          {\Gamma[\sigma]\vdash\phi[\sigma]}\)
      \hfill
      \(\displaystyle\inference{defineConst $c$}
          {t:\alpha}{c:\alpha\quad\vdash c = t}\)
      \hspace{2mm}
    }
    \vspace{5mm}
    \hbox to\hsize{
      \hspace{2mm}
      \(\displaystyle\inference{defineTypeOp $A[vs]$
            $\mbox{\textrm{\textit{abs}}}$
            $\mbox{\textrm{\textit{rep}}}$\footnotemark[4]}
          {t:\alpha\quad\vdash\phi\;t}{\vdash\mbox{\textrm{\textit{abs}}}\;
            (\mbox{\textrm{\textit{rep}}}\;a) = a\quad\vdash\phi\;r =
            (\mbox{\textrm{\textit{rep}}}\;
            (\mbox{\textrm{\textit{abs}}}\;r) = r)}\)
      \hfill
    }
    \vspace{2mm}
  }
}
\end{center}
\caption[The OpenTheory logical kernel]{The OpenTheory logical kernel.\footnotemark[5]}
\label{fig:otaxioms}
\end{figure}
\footnotetext[1]{The variable $v$ is assumed to be of type $\alpha$ in these rules.}
\footnotetext[2]{The terms $\phi,\phi'$ are required to be $\alpha$-equivalent.}
\footnotetext[3]{The variable $v$ must not be free in $\Gamma$. Equivalently, $\Gamma$ must be $\alpha$-equivalent to a set of terms $\Gamma'$ that do not have $v$ as a subterm.}
\footnotetext[4]{The list of variables $vs$ must match the set of free variables in $\phi$; this defines a type operator $A[vs]$ and constants $\mbox{\textrm{\textit{abs}}}:\alpha\to A[vs]$,
$\mbox{\textrm{\textit{rep}}}:A[vs]\to\alpha$.}
\footnotetext[5]{The astute reader will notice that this table is more detailed than figure 2 of \cite{otbook}, upon which it was based. This is due to the inclusion of explicit term formation axioms \textsf{varTerm}, \textsf{absTerm}, \textsf{appTerm}; these were elided in the presentation of \cite{otbook} but are included in later versions of the specification (see \url{http://gilith.com/research/opentheory/article.html}).}
}

An OpenTheory article file works as a stack machine, and the file format is very simple---a sequence of commands that manipulate the stack, separated by newlines \cite{otbook}. An article reader maintains a stack, a dictionary for backreferencing previous computations and theorems, and a list $\Gamma$ of assumptions and an export list $\Delta$ of theorems. The result of the computation after all instructions in the file are executed is a ``theory'' $\theory{\Gamma}{\Delta}$ that states that the theorems in $\Delta$ are derivable from the axioms in $\Gamma$. Figure~\ref{fig:otaxioms} shows the inference rules supported by the logical kernel. The base syntax involves two types of variables: $t,u,f,g,x,y$ denote term metavariables (and $\phi,\psi$ denote term metavariables of type $\sf bool$), while $\alpha,\beta$ are type metavariables. By ``metavariable'' we mean that an application of any of these axioms will not involve a literal $t$ but will have $t$ replaced by some term, which is itself constructed by application of these rules (by contrast to another notion of ``metavariable'' used in Metamath presentations; see section~\ref{sec:hol-mm}). The built-in constants are the type constant $\sf bool$, the function type operator $\alpha\to\beta$, and the equality operator ${=}:\alpha\to(\alpha\to{\sf bool})$ (the term $({=}\;x)\;y$ is presented as $x=y$ for clarity). Each term variable $v$ is constructed from a name (a string) and a type, and two variables are considered equal only if the name and type are equal.

A derivation is structured as a list of theorems of the form $\Gamma\vdash\phi$, where $\phi$ is a term of type $\sf bool$, and $\Gamma$ is a finite set of terms of type $\sf bool$. The set unions and differences in \textsf{deductAntisym} treat terms that are $\alpha$-equivalent as equal, where two terms are considered to be $\alpha$-equivalent if there is a consistent mapping of bound variables that transforms one term into the other. More precisely:

\begin{definition}
  Given a map $\sigma$ of variables to variables, terms $t,u$ are said to be \textit{$\alpha$-equivalent with respect to $\sigma$} if one of the following conditions is met:
  \begin{longitem}
    \item $t=u=c$ for some constant term $c$
    \item $t=u=v$ for some variable term $v$ not in the domain of $\sigma$
    \item $t=v$ and $u=\sigma(v)$ for some variable term $v$ in the domain of $\sigma$
    \item $t=f\;x$ and $u=g\;y$, and $f,g$ and $x,y$ are $\alpha$-equivalent with respect to $\sigma$
    \item $t=(\lambda w.\ x)$ and $u=(\lambda v.\ y)$, and $x,y$ are $\alpha$-equivalent with respect to $\sigma[w\mapsto v]$, where $\sigma[w\mapsto v]$ represents the map $\sigma$ with $\sigma(w)=v$ added to the map, replacing any mapping for $w$ if it exists.
  \end{longitem}
  Two terms $t,u$ are said to be \textit{$\alpha$-equivalent} if they are $\alpha$-equivalent with respect to an empty map.
\end{definition}

The term $\phi'$ in \textsf{eqMp} is required to be $\alpha$-equivalent to $\phi$. This definition is related to the definition of a proper substitution, used in rule \textsf{subst}.

\begin{definition}
  Given a map $\sigma$ of term variables to terms and a term $t$, the proper substitution $t[\sigma]$ is defined by structural induction as follows:
  \begin{longitem}
    \item For a constant term $t=c$ or a variable term $t=v$ not in the domain of $\sigma$, $t[\sigma]=t$.
    \item For a variable term $t=v$ in the domain of $\sigma$, $t[\sigma]=\sigma(v)$.
    \item If $t=f\;x$, then $t[\sigma]=f[\sigma]\;x[\sigma]$.
    \item If $t=(\lambda v.\ x)$ and some variable present as a subterm of $t$ (this includes $v$) is also in the domain of $\sigma$, then $t[\sigma]=(\lambda w.\ x[v\mapsto w][\sigma])$, where $w$ is a dummy variable distinct from variables in subterms of $t$ of the same type as $v$, otherwise $t[\sigma]=(\lambda v.\ x[\sigma])$.
  \end{longitem}
  For a set $\Gamma$ of terms, $\Gamma[\sigma]=\{t[\sigma]:t\in\Gamma\}$. We use the same notation $t[\sigma]$ for a substitution of \textit{type variables} to types, but in this case the substitution is direct (distributes through all term and type construction operators).
\end{definition}

The OpenTheory stack machine also contains a command \textsf{axiom} that accepts a list of terms $\Gamma$ and a term $\phi$ and produces the theorem $\Gamma\vdash\phi$, while also adding $\Gamma\vdash\phi$ to the list of assumptions, and a command \textsf{thm} that takes as input $\Gamma$ and $\phi$ and also a theorem $\Gamma'\vdash\phi'$ where $\Gamma'$ and $\phi'$ are $\alpha$-equivalent to $\Gamma$ and $\phi$, and adds $\Gamma\vdash\phi$ to the list of exported theorems. These rules are sufficient to completely describe the OpenTheory logical kernel.

In addition to the above axioms, there are three axioms that are appended in order to develop the full HOL foundations:

\label{dsc:otextras}
\begin{describe}{\textsf{extensionality}}
\item[\textsf{extensionality}] $\vdash\forall t.\ (\lambda x.\ t\; x)=t$
\item[\textsf{choice}] $\vdash\forall p, x.\ p\;x\Rightarrow p\;(\epsilon p)$
\item[\textsf{infinity}] $\vdash\exists f^{\sf ind\to ind}.\ ({\sf injective}\ f\wedge \neg{\sf surjective}\ f)$
\end{describe}

However, these require more advanced definitions such as $\forall$ and $\Rightarrow$, and their translations are straightforward, so can be safely ignored for Part I. These definitions also implicitly introduce the type constant $\sf ind$ (the intended model is as any infinite set) and the ``indefinite descriptor'' $\epsilon:(\alpha\to{\sf bool})\to\alpha$, which is the HOL equivalent of a global choice function. We will consider these constructs in more detail in Part II.

\subsection{\texttt{hol.mm}}
\label{sec:hol-mm}
  
The primary features of the OpenTheory system (see section~\ref{sec:opentheory}) that are not available in Metamath-based axiomatizations are proper substitution and $\alpha$-equivalence, the ``not free in'' predicate, and set manipulation of the contexts (the $\Gamma$ in $\Gamma\vdash\phi$). To address the last problem, we introduce a new term constructor, the ``context conjunction'', which takes as input terms $\phi$ and $\psi$ of type $\sf bool$ and produces a term $(\phi,\psi)$, also of type bool. When $\wedge$ is defined, it becomes possible to prove that $(\phi\wedge\psi)=(\phi,\psi)$, but before this it is necessary to introduce this as part of the axiomatization for the ``bootstrapping'' phase. After taking $\top$ as axiomatic, it becomes possible to represent $\Gamma\vdash\phi$ as simply $\psi\vdash\phi$ by using the context conjunction to put all the terms in $\Gamma$ into one term, and use $\psi=\top$ if $\Gamma=\emptyset$.

\afterpage{
\begin{figure}
\begin{center}
\framebox[\hsize][c]{
  \vbox{
    \vspace{2mm}
    \hbox to\hsize{
      \hspace{2mm}
      \(\displaystyle\inference{wtru}{}{\top:{\sf bool}}\)
      \hfill
      \(\displaystyle\inference{wct}{R:{\sf bool}\quad S:{\sf bool}}
          {(R,S):{\sf bool}}\)
      \hfill
      \(\displaystyle\inference{wc}{F:\alpha\to\beta\quad T:\alpha}
          {(F\;T):\beta}\)
      \hfill
      \(\displaystyle\inference{wv}{}
          {x^\alpha:\alpha}\)
      \hspace{2mm}
    }
    \vspace{5mm}
    \hbox to\hsize{
      \hspace{2mm}
      \(\displaystyle\inference{wl}{T:\beta}
          {(\lambda x^\alpha.\ T):\alpha\to\beta}\)
      \hfill
      \(\displaystyle\inference{id}{R:{\sf bool}}{R\vdash R}\)
      \hfill
      \(\displaystyle\inference{syl}{R\vdash S\quad S\vdash T}{R\vdash T}\)
      \hfill
      \(\displaystyle\inference{jca}
          {R\vdash S\quad R\vdash T}{R\vdash (S,T)}\)
      \hspace{2mm}
    }
    \vspace{5mm}
    \hbox to\hsize{
      \hspace{2mm}
      \(\displaystyle\inference{weq}{}{{=}:\alpha\to(\alpha\to{\sf bool})}\)
      \hfill
      \(\displaystyle\inference{simpl}{R:{\sf bool}\quad S:{\sf bool}}
          {(R,S)\vdash R}\)
      \hfill
      \(\displaystyle\inference{simpr}{R:{\sf bool}\quad S:{\sf bool}}
          {(R,S)\vdash S}\)
      \hspace{2mm}
    }
    \vspace{5mm}
    \hbox to\hsize{
      \hspace{2mm}
      \(\displaystyle\inference{trud}{R:{\sf bool}}{R\vdash\top}\)
      \hfill
      \(\displaystyle\inference{refl}{A:\alpha}{\top\vdash A = A}\)
      \hfill
      \(\displaystyle\inference{eqmp}{R\vdash A\quad R\vdash A = B}
          {R\vdash B}\)
      \hspace{2mm}
    }
    \vspace{5mm}
    \hbox to\hsize{
      \hspace{2mm}
      \(\displaystyle\inference{ded}{(R,S)\vdash T\quad (R,T)\vdash S}
          {R\vdash S = T}\)
      \hfill
      \(\displaystyle\inference{ceq}
          {F:\alpha\to\beta\quad A:\alpha\quad R\vdash F=G\quad R\vdash A=B}
          {R\vdash F\;A = G\;B}\)
      \hspace{2mm}
    }
    \vspace{5mm}
    \hbox to\hsize{
      \hspace{2mm}
      \(\displaystyle\inference{leq\footnotemark[1]}
          {R\vdash A=B}
          {R\vdash (\lambda x^\alpha.\ A)=(\lambda x^\alpha.\ B)}\)
      \hfill
      \(\displaystyle\inference{hbl1}{A:\gamma\quad B:\alpha}
          {\top\vdash(\lambda x^\alpha.\ \lambda x^\beta.\ A)\;B =
            (\lambda x^\beta.\ A)}\)
      \hspace{2mm}
    }
    \vspace{5mm}
    \hbox to\hsize{
      \hspace{2mm}
      \(\displaystyle\inference{distrc}{B:\alpha\quad F:\beta\to\gamma}
          {\top\vdash(\lambda x^\alpha.\ (F\;A))\;B =
            ((\lambda x^\alpha.\ F)\;B)\;((\lambda x^\alpha.\ A)\;B)}\)
      \hfill
      \(\displaystyle\inference{ax-17\footnotemark[2]}{A:\beta\quad B:\alpha}
          {\top\vdash(\lambda x^\alpha.\ A)\;B = A}\)
      \hspace{2mm}
    }
    \vspace{5mm}
    \hbox to\hsize{
      \hspace{2mm}
      \(\displaystyle\inference{distrl\footnotemark[3]}{A:\gamma\quad B:\alpha}
          {\top\vdash(\lambda x^\alpha.\ \lambda y^\beta.\ A)\;B =
            \lambda y^\beta.\ ((\lambda x^\alpha.\ A)\;B)}\)
      \hfill
      \(\displaystyle\inference{beta}{A:\beta}
          {\top\vdash(\lambda x^\alpha.\ A)\;x^\alpha = A}\)
      \hspace{2mm}
    }
    \vspace{5mm}
    \hbox to\hsize{
      \hspace{2mm}
      \(\displaystyle\inference{inst}{
           \begin{array}{c}
             \top\vdash(\lambda x^\alpha.\ B)\;y^\alpha = B\\
             \top\vdash(\lambda x^\alpha.\ S)\;y^\alpha = S\end{array}\quad
           \begin{array}{c}
             x^\alpha = C\vdash A = B\\
             x^\alpha = C\vdash R = S\end{array}
           \quad R\vdash A}{S\vdash B}\)
      \hspace{2mm}
    }
    \vspace{2mm}
  }
}
\end{center}
\caption{The \texttt{hol.mm} axiomatization.}
\label{fig:holmmaxioms}
\end{figure}
\footnotetext[1]{The variable $x$ is required to not be present in the (expression substituted for) $R$.}
\footnotetext[2]{The variable $x$ is required to not be present in $A$.}
\footnotetext[3]{The variable $y$ is required to be distinct from $x$ and not present in $B$.}
}

Figure~\ref{fig:holmmaxioms} shows the axiomatization that is used in \texttt{hol.mm}.\footnote{A version of \texttt{hol.mm} is available for download at \url{http://us.metamath.org/metamath/hol.mm}.} A comparison with the axioms of OpenTheory (figure~\ref{fig:otaxioms}) shows several differences. The most obvious difference is the increase in the number of axioms, from 14 to 23, but this comparison is deceptive, because this increase must be weighed against the more complicated explanation of $\alpha$-equivalence and proper substitution that must be described in order to fully explain when the axioms are applicable. In the \texttt{hol.mm} axiomatization, the only ``asterisks'' are the distinct variable provisos that come with \textsf{leq}, \textsf{ax-17}, and \textsf{distrl}.

There are three kinds of variables in these axioms: ``type'', ``var'', and ``term'' variables. Term variables are represented by capital letters, type variables are represented with greek letters, and ``var'' variables (referred to henceforth as \emph{vars}) are represented by lowercase letters. The reason for the division of OpenTheory term variables into two different types has to do with the way which Metamath handles substitutions. Metamath is designed to have ``pluggable'' axioms, with the only built-in axiom being the direct substitution of a (meta)variable with a term. This substitution process happens for every application of every theorem, and the only restriction on substitutions is that the substitution must be with a term, and the distinct variable provisos must be honored.

Unlike OpenTheory, vars in \texttt{hol.mm} do not come with built-in type information, so a variable term takes a var as well as a type; this is represented as a superscript in figure~\ref{fig:holmmaxioms} (e.g. $x^\alpha$). Metamath often refers to its variables as ``metavariables'', because they can be interpreted as variables which stand in for expressions of their type in some lower object language. Note that even vars like $x$ are metavariables in this sense, but they can only be substituted with other vars (because the only expressions of type ``var'' are other vars). The capital term variables, on the other hand, can be substituted with other terms, like $x^\alpha$ or $(\lambda x^\alpha.\ A)$. In Metamath, it is easier to work with term variables than vars because they admit direct substitution, but only vars can be bound in constructs like lambda abstractions, so both variable types are necessary in the Metamath interpretation. However, since this conflicts with our terminology we will stick to calling these variables, and referring to variables ranging over Metamath expressions as ``metavariables''. For simplicity of presentation, we will stick to using only vars during a derivation and keep the term variables to just the axioms themselves.

Most of the axioms have direct equivalents, but a few deserve extra explanation. The axioms \textsf{wtru}, \textsf{wct}, \textsf{id}, \textsf{syl}, \textsf{jca}, \textsf{simpl}, \textsf{simpr}, \textsf{trud} are necessary in order to handle context manipulation: doing the set unions in OpenTheory axioms like \textsf{eqMp} (this is discussed in more detail in section~\ref{sec:provingot}). The axioms \textsf{hbl1}, \textsf{distrc}, \textsf{distrl}, \textsf{ax-17} are used to define the properties of the ``not free in'' predicate, which is used in \textsf{inst}.

\begin{definition}
\label{def:nf}
  Let $\NF(x^\alpha,A)$ denote the statement that $\top\vdash(\lambda x^\alpha.\ A)\;y^\alpha = A$ is derivable (with explicit metavariable $y$). We can read this as ``$x$ is not free in $A$''.
\end{definition}

Then \textsf{ax-17} asserts that $\NF(x^\alpha,A)$ whenever $A$ is a term that is distinct from $x$, and \textsf{hbl1} asserts that $\NF(x^\alpha,(\lambda x^\beta.\ A))$, while \textsf{distrc} and \textsf{distrl} can be used to show that $\NF(x^\alpha,A)$ and $\NF(x^\alpha,B)$ imply $\NF(x^\alpha,A\;B)$ and $\NF(x^\alpha,(\lambda y^\beta.\ A))$ (when $x$ and $y$ are distinct). Thus this predicate allows one to represent the structural property ``$x$ is not free in $A$'' faithfully within the logic. (This same trick is used in \texttt{set.mm} to represent the not-free predicate, although there it is expressed more naturally as $(\varphi\to\forall x\,\varphi)$. The name \textsf{ax-17} is borrowed from \texttt{set.mm}, where an axiom by the same name asserts $(\varphi\to\forall x\,\varphi)$ when $x$ is distinct from $\varphi$.)

The three additional axioms in HOL are translated to the following three axioms:

\begin{describe}{\textsf{eta}}
\item[\textsf{eta}] $\top\vdash\forall f^{\alpha\to\beta}.\ (\lambda x^\alpha.\ f^{\alpha\to\beta}\; x^\alpha)=f^{\alpha\to\beta}$
\item[\textsf{ac}] $\top\vdash\forall p^{\alpha\to{\sf bool}}.\ \forall x^\alpha.\ (p^{\alpha\to{\sf bool}}\;x^\alpha\Rightarrow p^{\alpha\to{\sf bool}}\;(\epsilon p^{\alpha\to{\sf bool}}))$
\item[\textsf{inf}] $\top\vdash\exists f^{\sf ind\to ind}.\ ({\sf injective}\ f^{\sf ind\to ind}\wedge \neg{\sf surjective}\ f^{\sf ind\to ind})$
\end{describe}

These axioms are exactly what you would expect from the earlier listing on page~\pageref{dsc:otextras}. Also, the following simple theorems of the system will be useful:

$$\inference{a1i}{\top\vdash A}{R\vdash A}$$
$$\inference{eqcomi}{R\vdash A=B}{R\vdash B=A}$$
$$\inference{eqtri}{R\vdash A=B\quad R\vdash B=C}{R\vdash A=C}$$
$$\inference{cteq}{A:{\sf bool}\quad C:{\sf bool}\quad R\vdash A=B\quad R\vdash C=D}{R\vdash(A,C)=(B,D)}$$
$$\inference{hbct}{\NF(x^\alpha,A)\quad \NF(x^\alpha,B)}{\NF(x^\alpha,(A,B))}$$
$$\inference{cbv}{x^\alpha=y^\alpha\vdash A=B}{\top\vdash (\lambda x^\alpha.\ A)=(\lambda y^\alpha.\ B)}$$

In the last theorem, $y$ must not be present in $A$ and $x$ must not be present in $B$. This theorem is curious because although it is true using only OpenTheory core theorems (trivially, since $(\lambda x^\alpha.\ A)$ and $(\lambda y^\alpha.\ B)$ are $\alpha$-equivalent), it (provably) requires \textsf{eta} for its \texttt{hol.mm} proof.

\subsection{Language embedding}
\label{sec:embedding}

Our first step in defining the conversion from OpenTheory to \texttt{hol.mm} is to define the embedding of formulas in one language into formulas in the other. We denote this function as $\cal F$. 

\begin{definition}
  The map $\cal F$ operates on statements, terms, and types of the OpenTheory language, and outputs statements, terms, and types of the \texttt{hol.mm} language.
  \begin{longitem}
    \item For a constant term or type $c$, ${\cal F}(c)=c$, where a constant named $c$ is defined in the file if it is not already present.
    \item For a variable term $v$ of type $\alpha$, ${\cal F}(v)=v^{{\cal F}(\alpha)}$ and ${\cal F}(\lambda v.\ x)=(\lambda v^{{\cal F}(\alpha)}.\ {\cal F}(x))$, where a var named $v$ is defined in the file if it is not already present.
    \item For a type variable $\alpha$, ${\cal F}(\alpha)=\alpha$, where a type variable named $\alpha$ is defined in the file if it is not already present.
    \item ${\cal F}(f\;x)={\cal F}(f)\;{\cal F}(x)$
    \item ${\cal F}(t:\alpha)={\cal F}(t):{\cal F}(\alpha)$
    \item ${\cal F}(\alpha\to\beta)={\cal F}(\alpha)\to{\cal F}(\beta)$
    \item If $A[\alpha_1,\dots,\alpha_n]$ is a type operator of arity $n$, then ${\cal F}(A[\alpha_1,\dots,\alpha_n]) =$\\$A[{\cal F}(\alpha_1),\dots,{\cal F}(\alpha_n)],$ and $A[\beta_1,\dots,\beta_n]$, with literal type variables $\beta_1,\dots,\beta_n$, is added as a syntax constructor if it is not already present (and $\beta_1,\dots,\beta_n$ are added as type variables if not present).
    \item ${\cal F}(\{\phi_1,\phi_2,\dots,\phi_n\}\vdash\psi)=(({\cal F}(\phi_1),{\cal F}(\phi_2)),\dots,{\cal F}(\phi_n))\vdash{\cal F}(\psi)$, unless $n=0$ in which case ${\cal F}(\vdash\psi)=\top\vdash{\cal F}(\psi)$.
  \end{longitem}
\end{definition}

\begin{remark}
  Since sets are unordered but the context conjunction is, ${\cal F}(\Gamma\vdash\phi)$ is not uniquely defined. However, any of the choices of ordering of $\Gamma$ produce equivalent statements, and internally $\Gamma$ is usually stored as a list anyway, so one may as well use this ordering. Alternatively, one can define a total order on terms and insist that the listing be done in increasing order to ensure uniqueness. In any case, the ordering chosen will not be relevant to later developments.
\end{remark}

In the course of ``evaluating'' this function on the statements of an OpenTheory derivation, at various points certain variables and constants will be added to the logical system. This is necessary because all variable and constant names need to be predeclared in a Metamath file, so this ensures that the predefinitions are made and allows an OpenTheory file to use variables that may not have been defined in the \texttt{hol.mm} core (which only defines variables that are used in the axioms themselves, such as $x,y,A,B,\alpha,\beta$. Any other variable names, like $v$, will need to be declared before their use in a theorem).

We assume that all term variable names in an OpenTheory derivation are distinct from type variable names and \texttt{hol.mm} core axiom and theorem labels, and no variable is used multiple times in the same theorem statement with different types, because OpenTheory will consider these distinct while \texttt{hol.mm} will consider them the same (i.e. $\NF(x^\beta,(\lambda x^\alpha.\ x^\beta))$ even though OpenTheory would consider $x^\beta$ as free in that expression). This can be ensured with suitable preprocessing of the OpenTheory article file.

Now we are finally capable of stating the main goal of this part, although the proof will be postponed to the next section.

\begin{theorem}
\label{thm:hol-mm}
If the statement $\Gamma\vdash\phi$ is derivable in OpenTheory, then ${\cal F}(\Gamma\vdash\phi)$ is derivable in a conservative extension of {\normalfont\texttt{hol.mm}}.
\end{theorem}

The reason for the ``conservative extension'' caveat is because in addition to the variables and constants being added to the system by $\cal F$, our transformation will also need to add definitions coming from \textsf{defineConst} and \textsf{defineTypeOp}, and Metamath does not support a special definition construct. Instead, definitions and axioms are treated on equal footing, and an external tool can be used to show that the axioms that claim to be definitions are actually conservative.

\subsection{Proving the embedded OpenTheory axioms}
\label{sec:provingot}

Our proof of theorem \ref{thm:hol-mm} will proceed by defining an explicit map from OpenTheory derivations to \texttt{hol.mm} derivations. First, we show that proper substitution, simplification, and $\alpha$-equivalence are derivable.

\begin{lemma}
\label{thm:simp}
If $A$ is a nested context conjunction containing $B$ as a conjunct, then $A\vdash B$ is provable.
\end{lemma}
\begin{proof}
By induction on the length of $A$. If $A=B$, then \textsf{id} proves $A\vdash B$. If $A=(A_1,A_2)$ and $B$ is a conjunct of $A_1$, then \textsf{simpl} proves $A\vdash A_1$ and by induction $A_1\vdash B$ is provable, so \textsf{syl} proves $A\vdash B$. The case of $B$ a conjunct of $A_2$ is similar (using \textsf{simpr} instead).
\end{proof}

\begin{lemma}
\label{thm:eqthm}
If $A$ does not contain $x$ and $C=B[A/x]$ is the result of the proper substitution of $A$ for $x^\alpha$ in $B$ (where proper substitution of {\normalfont\texttt{hol.mm}} terms is defined similarly to OpenTheory proper substitution), then $x^\alpha=A\vdash B=C$ is provable, and $\NF(x^\alpha,C)$.
\end{lemma}
\begin{proof}
By induction on the length of $B$.
\begin{longitem}
\item If $B$ does not contain $x^\alpha$, then $C$ is identical to $B$ and so \textsf{refl}, \textsf{a1i} proves $x^\alpha=A\vdash B=B$ and \textsf{ax-17} proves $\NF(x^\alpha,B)$.
\item If $B=x^\alpha$, then $C=A$ so \textsf{id} proves $x^\alpha=A\vdash x^\alpha=A$ and \textsf{ax-17} proves $\NF(x^\alpha,A)$.
\item If $B=B_1\;B_2$, then $C=B_1[A/x]\;B_2[A/x]$, so the induction hypothesis gives proofs of $x^\alpha=A\vdash B_1=B_1[A/x]$ and $x^\alpha=A\vdash B_2=B_2[A/x]$, and \textsf{ceq} proves $x^\alpha=A\vdash B=C$ and \textsf{distrc}, \textsf{ceq} prove $\NF(x^\alpha,C)$ from $\NF(x^\alpha,B_1[A/x])$ and $\NF(x^\alpha,B_2[A/x])$. (The same is true when $B=(B_1,B_2)$, with \textsf{cteq} in place of \textsf{ceq} and \textsf{hbct} for the proof of  $\NF(x^\alpha,C)$.)
\item If $B=(\lambda x^\alpha.\ B_1)$, then $\NF(x^\alpha,B)$, so $C$ is identical to $B$ and so again \textsf{refl}, \textsf{a1i} proves $x^\alpha=A\vdash B=B$ and  $\NF(x^\alpha,B)$ is given.
\item If $B=(\lambda y^\beta.\ B_1)$ where $x,y$ are distinct, then $C=(\lambda y^\beta.\ B_1[A/x])$ and \textsf{distrl}, \textsf{leq}, \textsf{eqtri} proves $\NF(x^\alpha,C)$. If $A$ does not contain $y$, then \textsf{leq} proves the goal. If $A$ does contain $y$, then \textsf{leq} is not directly applicable, and the proper substitution for $C=(\lambda z^\beta.\ C_1)$ includes a dummy variable $z$. In this case, use the induction hypothesis to prove $y^\beta=z^\beta\vdash B_1=B_1[z/y]$, and then apply \textsf{cbv}, \textsf{a1i} to get $x^\alpha=A\vdash B=(\lambda z^\beta.\ B_1[z/y])$. Using the induction hypothesis once more to prove $x^\alpha=A\vdash B_1[z/y]=C_1$, \textsf{leq} gives $x^\alpha=A\vdash (\lambda z^\beta.\ B_1[z/y])=C$ and \textsf{eqtri} proves the goal theorem. 
\end{longitem}
\end{proof}

\begin{lemma}
\label{thm:alphaeq}
If $A$ is $\alpha$-equivalent to $B$ (where $\alpha$-equivalence of {\normalfont\texttt{hol.mm}} terms is defined similarly to OpenTheory $\alpha$-equivalence), then $\top\vdash A=B$ is provable.
\end{lemma}
\begin{proof}
By induction on the length of $A$.
\begin{longitem}
\item If $A$ is a variable or constant, then $A=B$ and \textsf{refl} proves the goal.
\item If $A=A_1\;A_2$, then $B=B_1\;B_2$ and \textsf{ceq} proves the goal.
\item If $A=(A_1,A_2)$, then $B=(B_1,B_2)$ and \textsf{cteq} proves the goal.
\item If $A=(\lambda x^\alpha.\ A_1)$ and $B=(\lambda x^\alpha.\ B_1)$, then \textsf{leq} proves the goal.
\item If $A=(\lambda x^\alpha.\ A_1)$ and $B=(\lambda y^\alpha.\ B_1)$, then let $A_1'=A_1[y/x]$ and $A'=(\lambda y^\alpha.\ A_1')$. Then $A$ is $\alpha$-equivalent to $A'$ and $B$, so by the third clause $\top\vdash A'=B$ is provable, and by lemma~\ref{thm:eqthm} $x^\alpha=y^\alpha\vdash A_1=A_1'$ is provable, so \textsf{cbv} gives $\top\vdash A=A'$ and \textsf{eqtri} proves the goal.
\end{longitem}
\end{proof}

\begin{lemma}
\label{thm:ae-sub} If $A$ is $\alpha$-equivalent to $B$ and $C$ is $\alpha$-equivalent to $D$, then $A\vdash C$ implies $B\vdash D$.
\end{lemma}
\begin{proof}
Lemma~\ref{thm:alphaeq} applied twice gives us $\top\vdash A=B$ and $\top\vdash C=D$, and \textsf{id}, \textsf{a1i}, \textsf{eqmp}, \textsf{eqcomi} turn these into $B\vdash A$ and $C\vdash D$, and then \textsf{syl} gives $B\vdash D$ as desired.
\end{proof}

In order to prove Theorem~\ref{thm:hol-mm}, we cast it as a special case of a more general theorem, using an invariant property which we'll call \textit{reduction}.

\begin{definition}
  Given an OpenTheory term $\phi$ and a \texttt{hol.mm} term $A$, we say that \textit{$A$ reduces to $\phi$} if there is a $\phi'$ $\alpha$-equivalent to $\phi$ with ${\cal F}(\phi')=A$, and given an OpenTheory statement $\Gamma\vdash\phi$ and a \texttt{hol.mm} statement $A\vdash B$, we say that \textit{$A\vdash B$ reduces to $\Gamma\vdash\phi$} if there is a \textit{type variable} substitution $\sigma$ such that $B$ reduces to $\phi[\sigma]$ and for every $\psi\in\Gamma$ either $A$ reduces to $\psi[\sigma]$ or $A=(A_1,A_2)$ and at least one of $A_1,A_2$ reduces to $\psi[\sigma]$.
\end{definition}

\begin{remark}
\label{rem:f-reduces}
  Note that ${\cal F}(\Gamma\vdash\phi)$ always reduces to $\Gamma\vdash\phi$.
\end{remark}

Intuitively, the notion of reduction from $A\vdash B$ to $\Gamma\vdash\phi$ means that $\Gamma$ is equivalent to a subset of the conjunction of terms in $A$, and $B$ and $\phi$ are equivalent. Thus if $\Gamma\vdash\phi$ is provable, then $A\vdash B$ has only added irrelevant antecedents, so it ought to be provable as well. This forms our main invariant across the derivation.

\begin{theorem}
\label{thm:reduction}
If $\Gamma\vdash\phi$ is derivable in OpenTheory and $A\vdash B$ reduces to $\Gamma\vdash\phi$, then $A\vdash B$ is derivable in a conservative extension of \normalfont{\texttt{hol.mm}}, and if $t:\alpha$ is derivable in OpenTheory, then ${\cal F}(t:\alpha)$ is derivable in a conservative extension of \normalfont{\texttt{hol.mm}}.
\end{theorem}

\begin{proof}
The proof is by induction on the length of the proof of the OpenTheory statement. We break the proof into cases based on the last inference rule in the derivation tree.

\subsubsection{Direct conversions}
Many of the axioms are converted directly into equivalent axioms. Specifically:

\begin{longitem}
\item $\sf varTerm\to wv$
\item $\sf absTerm\to wl$
\item $\sf appTerm\to wc$
\item $\sf refl\to refl$
\item $\sf eqMp\to eqmp$
\item $\sf appThm\to ceq$
\item $\sf deductAntisym\to ded$
\item $\sf extensionality\to eta$
\item $\sf choice\to ac$
\item $\sf infinity\to inf$
\end{longitem}

Given derivations of all the hypotheses to one of these inferences, apply the transformed step to the transformed hypotheses (and \textsf{a1i} to the result if $A\ne\top$) to get an $\alpha$-equivalent statement, and the lemma~\ref{thm:ae-sub} finishes the job.

\subsubsection{\textsf{assume}}
This axiom is an application of lemma~\ref{thm:simp} to prove $A\vdash B'$ where $B'$ is the conjunct of $A$ that is $\alpha$-equivalent to $B$, followed by lemma~\ref{thm:ae-sub}.

\subsubsection{\textsf{absThm}}
This axiom is almost a direct application of \textsf{leq}, but the requirement is only that $\Gamma$ not have $v$ free in it, not that $v$ be completely disjoint from $\Gamma$. However, if $v$ is not free in $\Gamma$, then there is a $\Gamma'$ that is disjoint from $v$ and $\alpha$-equivalent to $\Gamma$, so \textsf{leq} proves ${\cal F}({\Gamma'})\vdash(\lambda x^\alpha. A)=(\lambda x^\alpha. B)$ from ${\cal F}({\Gamma'})\vdash A=B$ and lemma~\ref{thm:ae-sub} applied before and after turn the $\Gamma'$ into $\Gamma$ in this inference.

\subsubsection{\textsf{subst}}
\label{sec:red-subst}
There are two kinds of substitution performed by \textsf{subst}---type variable substitution and term variable substitution. If a type variable substitution is performed, so that $\Gamma=\Gamma'[\sigma]$ and $\phi=\phi'[\sigma]$, then $\Gamma'\vdash\phi'$ is also reducible to $A\vdash B$, so $A\vdash B$ is provable.

If $\sigma=[x_1\mapsto A_1,\dots,x_n\mapsto A_n]$ is a term variable substitution, then by writing this as a composition of $[x_1\mapsto y_1,\dots,x_n\mapsto y_n]$ with  $[y_1\mapsto A_1,\dots,y_n\mapsto A_n]$ where $y_i$ are dummy variables, we can reduce this to the case when no $A_i$ contains any $x_j$. Then we can rewrite it again as the composition of $\sigma_1=[x_1\mapsto A_1],\dots,\sigma_n=[x_n\mapsto A_n]$, so that we can reduce to the case of a single variable substitution. And this case is handled by axiom \textsf{inst}, with the hypotheses filled by lemma~\ref{thm:eqthm}.

\subsubsection{\textsf{betaConv}}
This axiom is an application of \textsf{beta} followed by the same substitution process described in section~\ref{sec:red-subst} (and \textsf{a1i}, lemma~\ref{thm:ae-sub}).

\subsubsection{Definitions}
\label{sec:red-def}
Lastly, we have the two definitional axioms, \textsf{defineConst} and \textsf{defineTypeOp}. In this case, we simply introduce all the output statements as axioms. We can do a little better, though; by introducing the axioms 
$$\inference{eqtypri}{B:\alpha\quad R\vdash A=B}{A:\alpha}$$
$$\inference{typedef}{\top\vdash F\;B\quad \TD_{\beta,A,R}(F,B)}
  {A:\alpha\to\beta\quad R:\beta\to\alpha\quad
    \top\vdash(A;(R\;x^\beta)=x^\beta, F\;y^\alpha=(R\;(A\;y^\alpha)=y^\alpha))}$$
we can set it up so that the only axiom needed for \textsf{defineConst} is a single axiom $\top\vdash c=t$ for the constant's definition, and the only axiom needed to define a new type is $\TD_{\beta,A,R}(F,B)$ (which is a new kind of statement designed solely for input to \textsf{typedef}). However the checking that the constant does not appear in the definition, the typedef's type constant $\beta$ lists all free type variables in use, etc.\ must still be checked outside the system. Nonetheless, as long as the original OpenTheory derivation followed these consistency rules, the transformed definition will also be conservative, for the same reasons, so it does not interfere with this proof.
\end{proof}

\begin{proof}[Proof of Theorem \ref{thm:hol-mm}] Follows immediately from Remark~\ref{rem:f-reduces} and Theorem~\ref{thm:reduction}.
\end{proof} 

\section{Part II: Conversion from \MakeLowercase{\texttt{hol.mm}} to \MakeLowercase{\texttt{set.mm}}}
In this part, we have the remaining job of transforming our Metamath representation of a HOL axiomatic system into ZFC. Although the foundations are changing, the basic functions of substitution and the like are the same on both the start and endpoint of the transformation, so we can focus on the mathematical content of the sentences without worrying as much about the exact representation of the formula. We begin by describing the model of HOL that we will build in ZFC:

\begin{definition}
  A \textit{type} $\alpha$ is a pair $\langle\iota_\alpha,b_\alpha\rangle$ of a \textit{witness} and a \textit{base set} such that $b_\alpha\in V_{\omega+\omega}$ (where $V_{\omega+\omega}$ is the second limit step of the cumulative hierarchy) and $\iota_\alpha\in b_\alpha$. Let $\Ch(\epsilon)$ denote that either $\epsilon$ is a choice function on $V_{\omega+\omega}$ or there is no such function, and define a map ${\cal S}_\epsilon$ from types, terms and statements of \texttt{hol.mm} to sets and wffs of \texttt{set.mm}.
  \begin{longitem}
    \item For the two constant types, we take ${\cal S}_\epsilon({\sf bool})=\langle1,2\rangle:=\langle1,\{0,1\}\rangle$ and ${\cal S}_\epsilon({\sf ind})=\langle0,\omega\rangle$, and define ${\cal S}_\epsilon(\top)=1$.
    \item For type variables, ${\cal S}_\epsilon(\operatorname{type}\alpha)\leftrightarrow\alpha\in{\rm Type}$, where $\alpha\in{\rm Type}$ is defined to mean that $\alpha$ is a type in the sense above; this hypothesis is implicit in all axioms that involve type variables.
    \item For convenience, define the wff predicate $\toWff(A)\iff A=1$, and the function $\toBool(\varphi)=\operatorname{if}(\varphi,1,0)$.
    \item A variable is mapped to ${\cal S}_\epsilon(x^\alpha)=\operatorname{if}(x\in b_\alpha,x,\iota_\alpha)$. 
    \item For the function type, we take ${\cal S}_\epsilon(\alpha\to\beta) = \langle(x\in b_\alpha\mapsto\iota_\beta),b_\beta^{b_\alpha}\rangle$, the set of all ZFC functions from $\alpha$ to $\beta$, with a constant function as witness.  
    \item Function application is represented by function application:\\
    ${\cal S}_\epsilon(F\;A)=(F\,`\!A)$.
    \item Lambda abstraction is represented by the mapping operator\\
    ${\cal S}_\epsilon(\lambda x^\alpha.\ A)=(x\in b_\alpha\mapsto A)$.
    \item Context conjunction is represented by conjunction:\\
    ${\cal S}_\epsilon((A,B))=\toBool(\toWff(A)\wedge\toWff(B))$.
    \item For a term, ${\cal S}_\epsilon(A:\alpha)\leftrightarrow A\in b_\alpha$.
    \item For a theorem, ${\cal S}_\epsilon(A\vdash B)\leftrightarrow \vdash(\Ch(\epsilon)\wedge\toWff(A))\to\toWff(B)$.
    \item The equality operator $=_\alpha$ is mapped, depending on its type, to ${\cal S}_\epsilon(=_\alpha) = (x\in b_\alpha\mapsto(y\in b_\alpha\mapsto\toBool(x=y)))$
    \item The indefinite descriptor $\epsilon_\alpha$ itself is mapped, depending on its type, to
    $${\cal S}_\epsilon(\epsilon_\alpha) = (f\in b_\alpha^{\{0,1\}}\mapsto\operatorname{if}(\forall x\in b_\alpha.\ f(x)=0,\iota_\alpha,\epsilon(\{x\in b_\alpha:f(x)=1\}))).$$
    \item A defined type ${\cal S}_\epsilon(\TD_{\beta,A,R}(F,B))$ asserts that $\beta=\langle B,\{x\in b_\alpha:\toWff(F(x))\}\rangle$, $A=(x\in b_\alpha\mapsto \operatorname{if}(x\in b_\beta,x,B))$, and $R=(x\in b_\beta\mapsto x)$.
  \end{longitem}
\end{definition}

Most of these definitions are exactly what you would expect---functions are functions, and types and terms map to sets and their elements. The unusual part of the definition deals with the indefinite descriptor $\epsilon$. HOL is based on a version of the Axiom of Choice that is stronger than the usual one in ZFC. Instead of asserting that for any set there exists a choice function on that set, it asserts that a \textit{specific} function is a choice function on the universe. If the HOL universe were a proper class, this would be problematic, but luckily it can be entirely contained within $V_{\omega+\omega}$, which is a set in ZFC, and thus the ZFC axiom of choice gives us a single choice function $\epsilon$ on all of $V_{\omega+\omega}$. We pass this in as a parameter to $\cal S$, so that we can give meaning to the various pieces of the formula that use $\epsilon$, and theorems assert the choice behavior of $\epsilon$, so that it can be used in derivations. The reason for the ``or if there is no such function'' proviso is to allow the proof of $\vdash{\cal S}_\epsilon(\top\vdash A)\implies\vdash\toWff(A)$ to avoid choice, so that the ZFC choice axiom never gets invoked unless \textsf{ac} does. Since we also require a term variable to have unconditional closure, we are forced to add witnesses so that we don't need to invoke choice by using $\epsilon$ to select elements.

The construction here is performed using $V_{\omega+\omega}$, but of course it is also possible to use $V_\delta$ for any limit ordinal $\delta>\omega$, and this can be passed in as an extra parameter to make a translation ${\cal S}_{\delta,\epsilon}$ which depends on $\delta$. However, as this makes the translation process more cumbersome to describe, we will assume $\delta=\omega+\omega$ and leave the extension process to those who need the extra power this affords.

\subsection{Proving the \texttt{hol.mm} axioms}

\begin{theorem}
\label{thm:set-mm}
If $A\vdash B$ is derivable in \normalfont{\texttt{hol.mm}}, then $\vdash{\cal S}_\epsilon(A\vdash B)$ is derivable in \normalfont{\texttt{set.mm}}. In particular, if $A\vdash B$ does not contain $\epsilon$, then $\vdash{\cal S}(A)\to{\cal S}(B)$ is derivable (where the $\epsilon$ has been dropped from the notation to indicate that the action of ${\cal S}$ does not depend on $\epsilon$).
\end{theorem}

\begin{proof}
Here we merely need to verify that each of the axioms is preserved under wrapping by ${\cal S}_\epsilon$. Note that ${\cal S}_\epsilon$ can actually be defined in \texttt{set.mm}, so that each of the axioms, in ${\cal S}_\epsilon$-wrapped form, can be proven as theorems within \texttt{set.mm}, and then the transformation will be one-to-one in terms of proof length. Also, keep in mind that there are implicit hypotheses that $\alpha$, $\beta$, etc.\ are types; these become explicit hypotheses during this translation. Abbreviated proofs for each axiom are presented, using theorem labels for existing \texttt{set.mm} proofs when necessary \cite{setmm}.

\begin{longitem}
  \item $\operatorname{type}\,\sf bool$: $1\in 2\in V_{\omega+\omega}$
  \item $\operatorname{type}\,\sf ind$: $0\in \omega\in V_{\omega+\omega}$
  \item $\operatorname{type}\ (\alpha\to\beta)$: If $\operatorname{rank}(b_\alpha)=m<\omega+\omega$ and $\operatorname{rank}(b_\beta)=n<\omega+\omega$, then $\operatorname{rank}(b_\beta^{b_\alpha})\le \max(m,n)+3<\omega+\omega$
  \item \textsf{wtru}: $1\in 2$
  \item \textsf{wct}: $\toBool$ of anything is in $2$
  \item \textsf{wc}: If $f:\alpha\to\beta$ and $x\in\alpha$, then $f(x)\in\beta$
  \item \textsf{wv}: If $x\in b_\alpha$, then $\operatorname{if}(x\in b_\alpha,x,\iota_\alpha)=x\in b_\alpha$, otherwise $\iota_\alpha\in b_\alpha$
  \item \textsf{wl}: If for all $x$, $T\in b_\beta$, then $(x\in\alpha\mapsto T):b_\alpha\to b_\beta$
  \item \textsf{id}: theorem \textsf{simpr}\footnote{This and the other theorem names mentioned here refer to theorem labels in \texttt{set.mm}. There are individual web pages for these, for example \textsf{simpr} is found at \url{http://us.metamath.org/mpegif/simpr.html}.}
  \item \textsf{syl}: theorem \textsf{syldan}
  \item \textsf{jca}: theorem \textsf{jca}
  \item \textsf{weq}: use the definition and two applications of \textsf{wl} 
  \item \textsf{simpl,simpr}: theorem \textsf{simpl,simpr}
  \item \textsf{trud}: theorem \textsf{a1i,tru}
  \item \textsf{refl}: theorem \textsf{eqidd}
  \item \textsf{eqmp}: theorem \textsf{mpbid} after showing $\toWff(A=B)\leftrightarrow(\toWff(A)\leftrightarrow\toWff(B))$ by case analysis)
  \item \textsf{ded}: theorem \textsf{impbid,expr}
  \item \textsf{ceq}: theorem \textsf{fveq12d}
  \item \textsf{leq}: theorem \textsf{mpteq2dv}
  \item \textsf{hbl1}: theorem \textsf{hbmpt1}
  \item \textsf{distrc,distrl}: short proof using \textsf{vtoclg,fvmpt2}
  \item \textsf{ax-17}: theorem \textsf{fvmpt,eqidd}
  \item \textsf{beta}: theorem \textsf{fvmpt2}
  \item \textsf{inst}: theorem \textsf{vtoclf}
  \item \textsf{eta}: theorem \textsf{dffn5v}
  \item \textsf{inf}: short proof using $(x\in\omega\mapsto x+1)$
  \item \textsf{ac}: By definition, $\epsilon_\alpha(p)=\operatorname{if}(\forall x\in b_\alpha.\ p(x)=0,\iota_\alpha,\epsilon(\{x\in b_\alpha:p(x)=1\}))$, but since $p(x)=1$, the if-condition is false, so $\epsilon_\alpha(p)=\epsilon(\{x\in b_\alpha:p(x)=1\}))$. Assuming \textsf{ax-ac}, $\epsilon$ is a choice function on $V_{\omega+\omega}$, and $\{x\in b_\alpha:p(x)=1\}$ is a nonempty subset of $b_\alpha\in V_{\omega+\omega}$, so $\epsilon(\{x\in b_\alpha:p(x)=1\})\in \{x\in b_\alpha:p(x)=1\}$, and thus $p(\epsilon_\alpha(p))=1$.
\end{longitem}

For the final statement, observe that no ${\cal S}_\epsilon$ transformation depends on $\epsilon$ except ${\cal S}_\epsilon(\epsilon_\alpha)$, so if $\epsilon$ is not in $A$ or in $B$ then it will not be in the right hand side of $\Ch(\epsilon)\to({\cal S}(A)\to{\cal S}(B))$, and theorem \textsf{exlimiv} turns this into $\exists\epsilon.\Ch(\epsilon)\to({\cal S}(A)\to{\cal S}(B))$. But $\exists\epsilon.\Ch(\epsilon)$ is provable, because if there is a choice function $\epsilon$ on $V_{\omega+\omega}$ then $\Ch(\epsilon)$ for that choice of $\epsilon$ and if not then $\Ch(0)$ is true.
\end{proof}

\section{Future Work}
We stopped at Part II here, but one can argue the existence of a part III to this translation project, where notations such as the HOL Light natural numbers are mapped via the natural isomorphisms to the \texttt{set.mm} natural numbers, so that a statement like the Prime Number Theorem, which is (as of this writing) proven in HOL Light but not in \texttt{set.mm}, can be said to be proven in its ``natural form'', rather than in some model. This process is much less formulaic, however, and requires individual considerations of each mathematical concept in order to identify the proper isomorphisms.

\begin{acks}
The author wishes to thank Norman Megill, Bob Solovay, and Raph Levien for their work on investigating Part II in many email discussions, as well as John Harrison for creating HOL Light and Joe Leslie-Hurd for creating OpenTheory and providing the source material for figure~\ref{fig:otaxioms}.
\end{acks}

\begin{received}
Received February 20xx;
November 20xx;
accepted January 20xx
\end{received}

\nocite{*}
\bibliographystyle{alpha}
\bibliography{paper}

\end{document}